\documentclass[conference]{IEEEtran}
\usepackage{amsmath}
\usepackage{amsfonts}
\usepackage{epsfig}
\usepackage{amssymb}
\usepackage{cite}
\usepackage{multicol}
\usepackage{graphicx}
\usepackage{adjustbox}
\usepackage{xpatch,amsthm}
\usepackage{mathtools}
\usepackage{algorithm}
\usepackage{algpseudocode}
\usepackage{multicol,lipsum}
\usepackage{pifont}

\begin{document}
\onecolumn
\title{Capacity Region of the Symmetric Injective $K$-User Deterministic Interference Channel}
\author{\IEEEauthorblockN{Mehrdad Kiamari and A. Salman Avestimehr}\\
\IEEEauthorblockA{Department of Electrical Engineering, \\University of Southern California\\
Emails: kiamari@usc.edu and avestimehr@ee.usc.edu}}

\maketitle
\begin{abstract}
We characterize the capacity region of the symmetric injective $K$-user Deterministic Interference Channel (DIC) for all channel parameters.
The achievable rate region is derived by first projecting the achievable rate region of Han-Kobayashi (HK) scheme, which is in terms of common and private rates for each user, along the direction of aggregate rates for each user (i.e., the sum of common and private rates). We then show that the projected region is characterized by only the projection  of those facets in the HK region for which the coefficient of common rate and private rate are the same for all users, hence simplifying the region. Furthermore, we derive a tight converse for each facet of the simplified achievable rate region.
\end{abstract}
\section{Introduction} \label{sec1}

\IEEEPARstart{T}he deterministic interference channel (DIC), originally introduced in~\cite{elgamal}, represents a basic yet fruitful instance of interference channels that effectively captures the \emph{broadcast} and \emph{interference} phenomena in multi-user networks. For example, intuitions from the two-user DIC have lead to the capacity approximation of the two-user Gaussian interference channels in~\cite{etkin2008gaussian}. Further operational connections between the Gaussian interference channel and the two-user DIC are also established in~\cite{bresler2008two,avestimehr2015approximation}. However, despite its simplicity, characterizing the capacity region of the general $K$-user DIC has still remained an unsolved problem.

Our main result in this paper is to characterize the capacity region of the $K$-user DIC in a symmetric injective case. There have been several attempts at this problem in the past. In particular, the capacity region of the symmetric injective 3-user DIC has been characterized~\cite{jaafaar}. However, extending prior approaches to the general symmetric injective $K$-user DIC becomes extremely cumbersome due to the explosive growth in the number of parameters in both achievable schemes and the converse. To overcome this challenge, we propose new techniques in both the development of the achievable rate region and the converse.

For deriving the achievable rate region, we consider the general  Han-Kobayashi (HK) scheme~\cite{achiev}, in which the message of each user is into two parts: \emph{private message} which is supposed to be decoded only at the desired destination and \emph{common message} which is supposed to be decoded at all destinations. This scheme results in an achievable rate region that is in terms of \emph{common} and \emph{private} rates for the users. The challenge is then to eliminate the common and private rates and derive an achievable rate region that is in terms of the aggregate rates for the users (i.e., the sum of common and private rates). While a  Fourier-Moutzkin (FM) elimination method can be used to solve this problem for $K$-user DIC with small number of users (e.g. $K\leq 3$ as done in \cite{jaafaar}), applying FM method to networks with large number of users becomes extremely cumbersome. We overcome this challenge by directly projecting the achievable rate region of HK scheme along the direction of aggregate rates for the users, and exploiting the algebraic properties of the rate region to remove loose facets of the rate region. In particular, we show that the achievable rate region can be obtained by projecting \emph{only} those facets of achievable rate region of HK for which the coefficient of common rate and private rate are the same for all users.

We also derive a tight converse for each facet of the achievable rate region. In particular, we use the structure of the facets of the achievable rate region to systematically bound the mutual information between the transmit and receive signal of each user by the corresponding term in each facet.

The rest of the paper is organized as follows. We first explain the system model of the symmetric injective $K$-user DIC and state the main result in Section~\ref{sys_Thm}. We then elaborate upon the derivation of the achievable rate region in Section~\ref{ach}. Finally, in Section~\ref{converse}, we provide a tight converse for the achievable rate region.

\section{System Model and Main Result}\label{sys_Thm}
\newtheorem{rem}{Remark}

\begin{figure}
\centering
\includegraphics[trim =2.5in 4.05in 2.5in 0.5in, clip,width=0.4\textwidth]{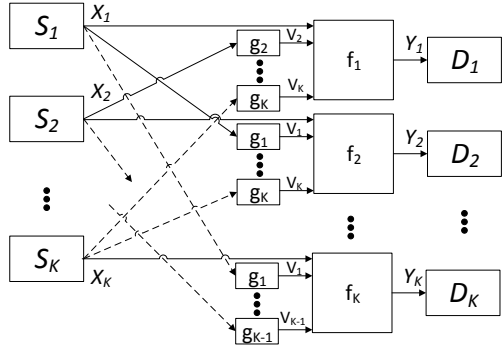}
\caption{The system model of the symmetric $K$-user DIC}
\label{sysmodel}
\end{figure}

The system model of the symmetric injective $K$-user DIC is shown in Fig.\ref{sysmodel}. In this model, source nodes and destination nodes are represented by $S_i$ and $D_i$ respectively, $\forall i\in [K]$ where $[K]\triangleq\{1,...,K\}$. Furthermore, the received signal at $D_i$, i.e. $Y_i$ is a deterministic function of transmitted signal $X_i$, where $X_i$ is a discrete random variable in finite set of alphabet $\mathcal X$, and interference signals $V_j$'s where $j\neq i$:
\begin{equation} \label{eq01a}
\begin{aligned}
V_i&=g_i(X_i),~~~~~~~~~~~~~~~~~~~~~~~~~~~~~~~\forall i \in [K],
\\Y_i&=f_i(X_i,V_{\Gamma_i})~~~\text{where}~{\Gamma_i}\triangleq[K]\backslash i,~~\forall i\in[K],
\end{aligned}
\end{equation}
for arbitrary function $g_i:\mathbb{R} \rightarrow \mathbb{R}$ and function $f_i:{\mathbb{R}}^{K} \rightarrow \mathbb{R}$,  $\forall i \in \{1,...,K\}$ such that the following conditions are satisfied
\begin{equation} \label{eq02a}
\begin{aligned}
H(Y_i|X_i)=H(V_{\Gamma_i})=\sum_{j\neq i}{H(V_j)}~~~~~~~\forall i\in [K],
\end{aligned}
\end{equation}
over all product distribution on $X_1,...,X_K$. This condition means $f_i$ is an invertible function given $X_i$, i.e.
\begin{equation} \label{eq02b}
\begin{aligned}
(V_1,V_2,...,V_{i-1},V_{i+1},...,V_K)=h_i(X_i,Y_i)~~\forall i \in [K],
\end{aligned}
\end{equation}
for some function $h_i(.)$, $\forall i\in[K]$. We refer to interference channel that satisfies this property as \emph{injective} deterministic interference channel.

\begin{rem}
\label{remark-sysmodel-1}
The invertibility restriction of $f_i(.)$s was originally imposed in \cite{elgamal} for the 2-user DIC and later generalized for the symmetric injective 3-user DIC in \cite{jaafaar}.
\end{rem}
\begin{rem}
\label{remark-sysmodel-2}
The above model is ``symmetric'' since transmitter $i$ causes the same interference $V_i=g_i(X_i)$ to all receivers.
\end{rem}

Based on the above model, we state our main result of this paper, which is the capacity characterization of the symmetric injective $K$-user DIC.

\newtheorem{theorem}{Theorem}
\begin{theorem}\label{theory-main}
The capacity region of the symmetric injective $K$-user deterministic interference channel is characterized as $\cup_{\prod_{i=1}^{K}{p(x_i)}}{\mathcal A}_{p(x_1),\dots,p(x_K)}$ where
\begin{equation} \label{eq03a}
\begin{aligned}
{\mathcal A}_{p(x_1),\dots,p(x_K)}&\triangleq \Big\{(R_1,...,R_K)|R_i \geq 0
 ,~~\textmd{$\forall i \in [K]$},~~ \sum_{i=1}^{K}{a_iR_i} \leq \sum_{i=1}^{K}{\sum_{j=1}^{a_i}{H(Y_i|V_{{\mathcal S}_{i,j}^c})}},\textmd{$a_i \in \mathbb{Z}^{\geq 0}$ and ${\mathcal S}_{i,j} \subseteq [K]$ },
\\&\text{satisfying}\sum_{i=1}^{K}{\sum_{j=1}^{a_i}{I_{{\mathcal S}_{i,j}}(m)=a_m}} ~~~\forall 1\leq m \leq K \Big\},
\end{aligned}
\end{equation}
and
\begin{equation} \label{eq0203}
\begin{aligned}
I_{{\mathcal S}_{i,j}}(m) &\triangleq \begin{cases} 1 &\mbox{if } m \in {\mathcal S}_{i,j} \\
0 & \mbox{otherwise }
\end{cases}
\\{\mathcal S}_{i,j}^c &\triangleq [K] \backslash {\mathcal S}_{i,j}.
\end{aligned}
\end{equation}
\end{theorem}

\begin{rem}
\label{remark1}
As a special case, Theorem \ref{theory-main} recovers the capacity region of the 2-user DIC of El Gamal-Costa \cite{elgamal}, by choosing $a_i$'s and ${\mathcal S}_{i,j}$'s in Table I. One can easily note that all other choices of $a_i$'s and ${\mathcal S}_{i,j}$'s not considered in the table result in redundant bounds.
\end{rem}
\begin{table}[ht]
\centering 
\begin{tabular}{|c|c|c|} 
\hline 
$a_i$'s & ${\mathcal S}_{i,j}$'s & Bound  \\ 
\hline 
$a_1=1$&${\mathcal S}_{1,1}=\{1\}$& (7) of \cite{elgamal}\\ 
$a_2=1$&${\mathcal S}_{2,1}=\{2\}$& (8) of \cite{elgamal}\\ 
$a_1=1,a_2=1$&${\mathcal S}_{1,1}=\{\},{\mathcal S}_{2,1}=\{1,2\}$& (9.a) of \cite{elgamal}\\
$a_1=1,a_2=1$&${\mathcal S}_{1,1}=\{1,2\},{\mathcal S}_{2,1}=\{\}$& (9.b) of \cite{elgamal}\\ 
$a_1=1,a_2=1$&${\mathcal S}_{1,1}=\{2\},{\mathcal S}_{2,1}=\{1\}$ &  (9.c) of \cite{elgamal}\\ 
$a_1=2,a_2=1$&${\mathcal S}_{1,1}=\{1,2\},{\mathcal S}_{1,2}=\{\},{\mathcal S}_{2,1}=\{1\}$& (10) of \cite{elgamal}\\ 
$a_1=1,a_2=2$&${\mathcal S}_{1,1}=\{2\},{\mathcal S}_{2,1}=\{1,2\},{\mathcal S}_{2,2}=\{\}$& (11) of \cite{elgamal}\\ 
\hline 
\end{tabular}
\label{table-2user} 
\caption{Table I: $a_i$'s and ${\mathcal S}_{i,j}$'s for the case $K=2$}
\end{table}

\begin{rem}
\label{remark2}
As a special case, Theorem \ref{theory-main} also recovers the capacity region of the symmetric injective 3-user DIC \cite{jaafaar}. The choice of $a_i$s and ${\mathcal S}_{i,j}$ are illustrated in Table II in Appendix C.
\end{rem}
\begin{rem}
\label{converse_gcs_thm1}
In \cite{allerton_gcs} it was shown that for the 2-user and 3-user DIC, a tight converse can be derived by applying the Generalized Cut-Set (GCS) bound \cite{gcs} on an appropriately designed ``extended network''. To develop the converse for Theorem 1, we instead use a direct method to systematically bound the mutual information between the transmit and receive signal of each user by the corresponding term in each facet of (\ref{eq03a}).
\end{rem}

\section{Achievability}\label{ach}
To prove Theorem \ref{theory-main}, we fix the product distribution $\prod_{i=1}^{K}{p(x_i)}$ and show that the region $\mathcal A\triangleq {{\mathcal A}_{p(x_1),\dots,p(x_K)}}$ is achievable.
We first start by considering the achievable rate region of HK scheme as described below.

\noindent \textbf{Codeword Generation:} Consider a product distribution ${\prod_{i=1}^K{p_{X_i}(x_i)}}$ and the corresponding $\prod_{i=1}^K{p_{V_i}(v_i)}$ (note that $V_i=g_i(X_i)$). $2^{nR_{ic}}$ independent codewords of length $n$, denoted by $V_{i}^n(c_i)$ where $c_i \in \{1,2,...,2^{nR_{ic}}\}$, are generated according to $\prod_{k=1}^n{p_{V_i}(v_{ik})}$ by transmitter $i$ for $i=1,...,K$. Then, $2^{nR_{ip}}$ independent codewords of length $n$, denoted by $X_{i}^n(c_i,p_i)$, are generated according to $\prod_{k=1}^n{p_{X_i|V_i}(x_{ik}|v_{ik})}$ for each codeword $V_{i}^n(c_i)$. Now, transmitter $i$ sends $X_{i}^n(c_i,p_i)$ as the corresponding coded signal of message with index $(c_i,p_i)$.

\noindent \textbf{Decoding:} For $i=1,...,K$, receiver $i$ tries to find a unique $(\hat {c_i},\hat {p_i})$ and $(\hat {c_1},...,\hat {c_{i-1}},\hat {c_{i+1}},...,\hat {c_K})$ satisfying
\small
\begin{equation} \label{eq04a}
\begin{aligned}
(X_{i}^n(\hat {c_i},\hat {p_i}),V_{1}^n(\hat {c_1}),...,V_{K}^n(\hat {c_K}),Y_i^n) \in A_{n}^{\epsilon}(X_i,V_1,..,V_K,Y_i).
\end{aligned}
\end{equation}
\normalsize

\noindent \textbf{Error Probability:} As we show in Appendix \ref{error_com_priv}, the error probability for all receivers goes to zero as $n$ increases for any rate tuple $(R_{1p},R_{1c},...,R_{Kp},R_{Kc})$ that is in the following region.
\small
\begin{equation} \label{eq05a}
\begin{aligned}
{\mathcal A}_1 \triangleq \Big\{&(R_{1p},R_{1c},...,R_{Kp},R_{Kc})|R_{ip},R_{ic} \geq 0, ~R_{ip}+\sum_{k \in {\mathcal M}_{i,j}}{R_{kc}}\leq H(Y_i|V_{{\mathcal M}_{i,j}^c})~~\text{for all},
~\text{${\mathcal M}_{i,j} \subseteq [K]$, $\forall j\in [2^K]$, and $\forall i\in [K]$}\Big\},
\end{aligned}
\end{equation}
\normalsize
where ${\mathcal M}_{i,j}$ is the $j$th subset of $\{1,...,K\}$ for the corresponding $i$.

Hence, by considering the aggregate  rates  for  each  user (i.e.,  the  sum  of  common  and  private  rates), we achieve the following rate region.
\small
\begin{equation} \label{eqa1}
\begin{aligned}
{\mathcal A}_2&\triangleq \Big\{{(R_1,...,R_K)|R_i=R_{ic}+R_{ip}~\text{, $~~ \forall i\in [K]$}}
~\text{where}~ (R_{1p},R_{1c},...,R_{Kp},R_{Kc}) \in {\mathcal A}_1\Big\}.
\end{aligned}
\end{equation}
\normalsize
 It is clear that region ${\mathcal A}_2$ can be obtained by projecting region ${\mathcal A}_1$ according to the following linear transformation matrix ${\bf A}_{K \times 2K}$
where
\begin{equation} \label{eqloosebound}
\begin{aligned}
[{\bf A}]_{ij} =
\left\{
	\begin{array}{ll}
		1  & \mbox{if } j=2i~\text{or}~j=2i-1, \\
		0 & \mbox{Otherwise.}
	\end{array}
\right.
\end{aligned}
\end{equation}

Based on this projection, we now claim that the achievable rate region  ${\mathcal A}_2$ can be characterized as the following.
\newtheorem{lem}{Lemma}
\begin{lem}
\label{lemext}
Region ${\mathcal A}_2$ is equivalent to following region
\small
\begin{equation} \label{eqa2}
\begin{aligned}
{\mathcal A}_3 \triangleq \Big\{ &(R_1,...,R_K)|\sum_{i=1}^K{a_iR_i}\leq \sum_{i=1}^K{\sum_{j=1}^{2^K}{c_{i,j}H(Y_i|V_{{\mathcal M}_{i,j}^c})}}, R_i \geq 0~\text{for all $i \in \{1,...,K\}$, $c_{i,j} \in \mathbb{Z}^{\geq 0}$, and ${\mathcal M}_{i,j} \subseteq [K]$},
\\&\text{satisfying}~ a_m=\sum_{j=1}^{2^K}{c_{m,j}}=\sum_{i=1}^{K}{\sum_{j=1}^{2^K}{c_{i,j}I_{{\mathcal M}_{i,j}}(m)}}, \forall m\in [K] \Big\}.
\end{aligned}
\end{equation}
\normalsize
\end{lem}
\begin{proof}[{Proof}]
As mentioned earlier, region ${\mathcal A_2}$ can be obtained by projecting region ${\mathcal A}_1$ based on linear transformation $\bf A$. Since region ${\mathcal A}_1$ is polyhedra, region ${\mathcal A_2}$ can be found by projecting all facets of ${\mathcal A}_1$ according to the projection matrix $\bf A$.

Note that the facets of region ${\mathcal A}_1$ are obtained by linear combinations of the inequalities characterizing this region. Hence, according to (\ref{eq05a}), all possible facets of ${\mathcal A}_1$ can be written as follows.
\begin{equation} \label{general}
\begin{aligned}
&\sum_{i=1}^K{\sum_{j=1}^{2^K}{{c}_{i,j}(R_{ip}+\sum_{k \in {\mathcal M}_{i,j}}{R_{kc}})}}~~\leq \sum_{i=1}^K{\sum_{j=1}^{2^K}{{c}_{i,j}H(Y_i|V_{{\mathcal M}_{i,j}^c})}},
\end{aligned}
\end{equation}
for all ${\mathcal M}_{i,j} \subseteq [K]$ and ${c}_{i,j}\in {\mathbb{R}}_{\geq 0}$  ($i=1,...,K$ and $j=1,...,2^K$), where ${c}_{i,j}$ is the corresponding coefficient of the inequality in (\ref{eq05a}) with subset ${\mathcal M}_{i,j}$.


%

By simplifying the left-hand side of (\ref{general}), we have
\begin{equation} \label{generalnnw}
\begin{aligned}
&\sum_{i=1}^K{\sum_{j=1}^{2^K}{{c}_{i,j}(R_{ip}+\sum_{k \in {\mathcal M}_{i,j}}{R_{kc}})}}
\\&=\sum_{i=1}^K{\sum_{j=1}^{2^K}{{c}_{i,j}R_{ip}}}+\sum_{i=1}^K \sum_{j=1}^{2^K}  \sum_{k=1}^K {c}_{i,j} I_{{\mathcal M}_{i,j}}(k)  R_{kc}
\\&=\sum_{i=1}^K{\sum_{j=1}^{2^K}{{c}_{i,j}R_{ip}}}+ \sum_{k=1}^K \sum_{i=1}^K \sum_{j=1}^{2^K}  {c}_{i,j} I_{{\mathcal M}_{i,j}}(k)  R_{kc}
\\&=\sum_{m=1}^K{\sum_{j=1}^{2^K}{{c}_{m,j}R_{mp}}}+ \sum_{m=1}^K \sum_{i=1}^K \sum_{j=1}^{2^K}  {c}_{i,j} I_{{\mathcal M}_{i,j}}(m)  R_{mc}
 \\&=\sum_{m=1}^K{d_mR_{mp}}+\sum_{m=1}^K{e_mR_{mc}},
\end{aligned}
\end{equation}
where
\begin{equation} \label{eq2conditions2}
\begin{aligned}
d_m&\triangleq\sum_{j=1}^{2^K}{c_{m,j}}~~,~~e_m&\triangleq \sum_{i=1}^{K}{\sum_{j=1}^{2^K}{c_{i,j}I_{{\mathcal M}_{i,j}}(m)}},~~\forall m\in [K].
\end{aligned}
\end{equation}

Thus, all facets of region ${\mathcal A}_1$  can now be written as
\begin{equation} \label{general2}
\begin{aligned}
&\sum_{i=1}^K d_i R_{ip} + \sum_{i=1}^K e_i R_{ic}  ~~\leq \sum_{i=1}^K{\sum_{j=1}^{2^K}{{c}_{i,j}H(Y_i|V_{{\mathcal M}_{i,j}^c})}},
\end{aligned}
\end{equation}
for all ${\mathcal M}_{i,j} \subseteq [K]$, ${c}_{i,j}\in {\mathbb{R}}_{\geq 0}$, and $d_i$ and $e_i$ defined in (\ref{eq2conditions2}). Now note that, according to Lemma \ref{proj_min} proved in Appendix \ref{app_proj_min}, the projection of (\ref{general2}) according to the linear transformation matrix $\bf A$ would result in the following bound
\begin{equation} \label{generalproj}
\begin{aligned}
&\sum_{i=1}^K{\min(d_i,e_i)R_{i}}~~\leq \sum_{i=1}^K{\sum_{j=1}^{2^K}{{c}_{i,j}H(Y_i|V_{{\mathcal M}_{i,j}^c})}}.
\end{aligned}
\end{equation}
Therefore, region ${\mathcal A_2}$ that is obtained by projection of region ${\mathcal A}_1$ based on linear transformation $\bf A$, is characterized as
\begin{equation} \label{generalproj2}
\begin{aligned}
\Big\{&(R_1,...,R_K)|\sum_{i=1}^K{a_iR_i}\leq \sum_{i=1}^K{\sum_{j=1}^{2^K}{c_{i,j}H(Y_i|V_{{\mathcal M}_{i,j}^c})}}, R_i \geq 0
~\text{for all $i \in \{1,...,K\}$, $c_{i,j} \in \mathbb{Z}^{\geq 0}$,  ${\mathcal M}_{i,j} \subseteq [K]$},
\\&\text{and $a_i=\min(d_i,e_i)$, where $d_i$ and $e_i$  are defined in (\ref{eq2conditions2})} \Big\}.
\end{aligned}
\end{equation}

Now, note that region ${\mathcal A}_3$ is exactly the same as the above region, except further restricting to the choice of $c_{i,j}$'s to satisfy
\begin{equation}
\label{eq:equalityConstraint}
a_m=d_m=e_m, ~\forall m\in [K].
\end{equation}
To complete the proof, we only need to show that any inequality in (\ref{generalproj2}) that does not satisfy the constraint of (\ref{eq:equalityConstraint}) is redundant, meaning that it can be obtained by linear combination of inequalities already considered in region ${\mathcal A}_3$. Let us consider inequalities that can be written as follows
\begin{equation} \label{ref_non_eq_coef}
\begin{aligned}
&\sum_{i=1}^K{\sum_{j=1}^{2^K}{{c}_{i,j}(R_{ip}+\sum_{k \in {\mathcal M}_{i,j}}{R_{kc}})}}~~\leq \sum_{i=1}^K{\sum_{j=1}^{2^K}{{c}_{i,j}H(Y_i|V_{{\mathcal M}_{i,j}^c})}},
\end{aligned}
\end{equation}
where
\begin{equation} \label{eq2conditionsoooo}
\begin{aligned}
\sum_{j=1}^{2^K}{c_{m,j}}&\neq\sum_{i=1}^{K}{\sum_{j=1}^{2^K}{c_{i,j}I_{{\mathcal M}_{i,j}}(m)}}&~~~\text{for some}~~m.
\end{aligned}
\end{equation}

We now demonstrate the following two consecutive steps to find an inequality in region ${\mathcal A}_2$ such that it satisfies (\ref{eq:equalityConstraint}) and its projection results in a tighter bound than the projection of (\ref{ref_non_eq_coef}). In the first step, we find an inequality such that the coefficient of private rate is greater than or equal to the coefficient of common rate for all users and its projection according to transformation matrix $\bf A$ results in a tighter bound on $\sum_{k=1}^K{\min(d_k,e_k)R_{k}}$ compared to the projection of (\ref{ref_non_eq_coef}).  In the second step, we obtain another inequality, based on the resulting inequality of the first step, such that the coefficient of private rate and common rate are the same for all users while its projection leads to a tighter bound on $\sum_{k=1}^K{\min(d_k,e_k)R_{k}}$ compared to the projection of (\ref{ref_non_eq_coef}).
\\

\noindent{\bf Step 1:} In this step, we aim to find new inequality
\small
\begin{equation} \label{ach_1}
\begin{aligned}
&\sum_{i=1}^K{\sum_{j=1}^{2^K}{{\tilde c}_{i,j}(R_{ip}+\sum_{k \in {\mathcal M}_{i,j}}{R_{kc}})}}~~\leq \sum_{i=1}^K{\sum_{j=1}^{2^K}{{\tilde c}_{i,j}H(Y_i|V_{{\mathcal M}_{i,j}^c})}},
\end{aligned}
\end{equation}
\normalsize
such that $a)$
\small
\begin{equation} \label{condition_big_ineq_a}
\begin{aligned}
{\tilde d}_{m}&={d}_{m}~~,~~{\tilde e}_{m}&= \min{(e_{m},d_{m})}~,~\forall m \in [K],
\end{aligned}
\end{equation}
\normalsize
where ${\tilde d}_{m}$ and ${\tilde e}_{m}$ are similarly defined as (\ref{eq2conditions2}) and $b)$ the projection of (\ref{ach_1}), according to transformation matrix ${\bf A}$, results in a tighter bound on $\sum_{k=1}^K{\min(d_k,e_k)R_{k}}$ compared to the projection of (\ref{ref_non_eq_coef}).

Let us consider a $m$ where $e_m>d_m$ and define
\small
\begin{equation} \label{def_JNa}
\begin{aligned}
~~~{\mathcal J}_c &\triangleq \{(i,j)|i \in [K]\backslash m,j\in[2^K],~\text{and}~m\in {\mathcal M}_{i,j}\},\\
~~~{\mathcal N}_c &\triangleq \{(i,j')|{\mathcal M}_{i,j'}={\mathcal M}_{i,j}\backslash m, \forall(i,j)\in {\mathcal J}_c \},\\
\end{aligned}
\end{equation}
\normalsize
\small
\begin{equation} \label{def_alpha_beta_c}
\begin{aligned}
~~~~~~~~~~~~~{\alpha}_{i,j'}\triangleq {\beta}_{i,j}~~\text{for $(i,j')\in {\mathcal N}_c$, $(i,j) \in {\mathcal J}_c$ s.t. ${\mathcal M}_{i,j'}={\mathcal M}_{i,j}\backslash m$},
\end{aligned}
\end{equation}
\normalsize
for all $ {\beta}_{i,j} \in \mathbb{Z}^{\geq 0}$ satisfying ${\beta}_{i,j}\leq {c}_{i,j}$, $\forall(i,j)\in {\mathcal J}_c$ and
\small
\begin{equation} \label{def_sum_c}
\begin{aligned}
\sum_{(i,j)\in {\mathcal J}_c} {{\beta}_{i,j}} &=e_m-d_m ,\\
{\beta}_{i,j}&=0~,~ \forall(i,j)\not\in {\mathcal J}_c.
\end{aligned}
\end{equation}
\normalsize
By introducing coefficient ${\hat c}_{i,j}$ as follows
\small
\begin{equation} \label{def_Cija}
\begin{aligned}
{\hat c}_{i,j}\triangleq
\left\{
	\begin{array}{ll}
		c_{i,j}  & \mbox{if } (i,j)\not\in {\mathcal N}_c~,~(i,j)\not \in {\mathcal J}_c, \\
		c_{i,j}+{\alpha}_{i,j} & \mbox{if } (i,j)\in {\mathcal N}_c,\\
		c_{i,j}-{\beta}_{i,j}  & \mbox{if } (i,j)\in {\mathcal J}_c,
	\end{array}
\right.
\end{aligned}
\end{equation}
\normalsize
we now consider the following inequality in region ${\mathcal A}_2$
\small
\begin{equation} \label{ineq_a}
\begin{aligned}
\sum_{i=1}^K{\sum_{j=1}^{2^K}{{\hat c}_{i,j}(R_{ip}+\sum_{k \in {\mathcal M}_{i,j}}{R_{kc}})}}~~\leq \sum_{i=1}^K{\sum_{j=1}^{2^K}{{\hat c}_{i,j}H(Y_i|V_{{\mathcal M}_{i,j}^c})}},
\end{aligned}
\end{equation}
\normalsize
and show that
\small
\begin{equation} \label{condition_small_ineq_a}
\begin{aligned}
(A)\quad
\left\{
	\begin{array}{ll}
{\hat d}_{m}={d}_{m}~&,~{\hat e}_{m}= \min{(e_{m},d_{m})}\\{\hat d}_{m'}={d}_{m'}~&,~{\hat e}_{m'}= e_{m'}~~~~\forall m'\neq m,
\end{array}
\right.
\end{aligned}
\end{equation}
\normalsize
\small
\begin{equation} \label{condition_small_ineq_a_B}
\begin{aligned}
~~~~~~~~~~~~~(B)\quad
\left\{
	\begin{array}{ll}
\sum \limits_{\substack{(i,j):\\i\in[K],j\in[2^K]}}{\hspace{-4mm}{\hat c}_{i,j}H(Y_i|V_{{\mathcal M}_{i,j}^c})} \leq \hspace{-4mm}{\sum \limits_{\substack{(i,j):\\i\in[K],j\in[2^K]}}{\hspace{-4mm}{c}_{i,j}H(Y_i|V_{{\mathcal M}_{i,j}^c})}}.
\end{array}
\right.
\end{aligned}
\end{equation}
\normalsize

We first prove the existence of such $\beta_{i,j}$'s satisfying (\ref{def_sum_c}) in Claim \ref{exist_claim_c}.
\newtheorem{claim}{Claim}
\begin{claim}\label{exist_claim_c}
By considering (\ref{def_JNa}), there exists ${\beta}_{i,j}\leq {c}_{i,j}$ for all $(i,j)$ such that (\ref{def_sum_c}).
\end{claim}
\newtheorem{claimproof}{Proof of claim}
\begin{proof}
One can easily verify the existence of such ${\beta}_{i,j}\leq {c}_{i,j}$, $\forall(i,j)\in {\mathcal J}_c$ satisfying (\ref{def_sum_c}) by showing that $\sum_{(i,j)\in {\mathcal J}_c} {{\beta}_{i,j}}$ can take value $e_m-d_m$ as follows
\small
\begin{equation} \label{lem_exist_c}
\begin{aligned}
\sum_{(i,j)\in {\mathcal J}_c} {{\beta}_{i,j}} &\leq \sum_{(i,j)\in {\mathcal J}_c} {{c}_{i,j}} \\&= \sum_{(i,j):m \in {\mathcal M}_{i,j}} {{c}_{i,j}}
-\sum_{j:m \in {\mathcal M}_{m,j}} {{c}_{m,j}}
\\&=\underbrace{\sum_{(i,j):m \in {\mathcal M}_{i,j}} {\hspace{-3mm}{c}_{i,j}}}_{=e_m}+\underbrace{\sum_{j:m \not\in {\mathcal M}_{m,j}} {\hspace{-3mm}{c}_{m,j}}}_{\geq 0}-\underbrace{(\sum_{j:m \not\in {\mathcal M}_{m,j}} {\hspace{-3mm}{c}_{m,j}}+\hspace{-3mm}\sum_{j:m \in {\mathcal M}_{m,j}} {\hspace{-3mm}{c}_{m,j}})}_{=d_m}.
\end{aligned}
\end{equation}
\normalsize
\end{proof}
We next present Claim \ref{coeff_claim_c} to verify (\ref{condition_small_ineq_a}).
\begin{claim} \label{coeff_claim_c}
By considering (\ref{def_JNa})-(\ref{def_Cija}), we have (\ref{condition_small_ineq_a}).
\end{claim}
\begin{proof}
We first find ${\hat d}_{m}$ and ${\hat e}_{m}$, then we derive ${\hat d}_{m'}$ and ${\hat e}_{m'}$ for all $m'\neq m$ as follows
\small
\begin{equation} \label{def_d_m_c}
\begin{aligned}
{\hat d}_m=\sum_{j=1}^{2^K}{{\hat c}_{m,j}}\overset{(a)}{=}\sum_{j=1}^{2^K}{c_{m,j}}=d_m,
\end{aligned}
\end{equation}
\normalsize
where step $(a)$ follows from the fact $(m,j)\not\in {\mathcal N}_c$ and $(m,j)\not\in {\mathcal J}_c$.
\small
\begin{equation} \label{def_e_m_c}
\begin{aligned}
{\hat e}_{m}&=\sum_{i=1}^K{\sum_{j=1}^{2^K}{{\hat c}_{i,j}I_{{\mathcal M}_{i,j}}(m)}}\\&=\sum_{\substack{(i,j):(i,j)\not\in{\mathcal N}_c, (i,j)\not\in{\mathcal J}_c}}{\hspace{-7mm}{\hat c}_{i,j}I_{{\mathcal M}_{i,j}}(m)}
+\sum_{\substack{(i,j):(i,j)\in{\mathcal N}_c}}{\hspace{-3mm}{\hat c}_{i,j}I_{{\mathcal M}_{i,j}}(m)}+
\sum_{\substack{(i,j):(i,j)\in{\mathcal J}_c}}{\hspace{-3mm}{\hat c}_{i,j}I_{{\mathcal M}_{i,j}}(m)}
\\&\overset{(a)}{=}\hspace{-4mm}\sum_{\substack{(i,j):i\in [K],j\in[2^K]}}{\hspace{-3mm}{c}_{i,j}I_{{\mathcal M}_{i,j}}(m)}+\hspace{-4mm}\sum_{\substack{(i,j'):(i,j')\in{\mathcal N}_c}}{\hspace{-6mm}{\alpha}_{i,j'}I_{{\mathcal M}_{i,j'}}(m)}
-\hspace{-4mm}\sum_{\substack{(i,j):(i,j)\in{\mathcal J}_c}}{\hspace{-3mm}{\beta}_{i,j}I_{{\mathcal M}_{i,j}}(m)}
\\&\overset{(b)}{=} \hspace{-6mm}\sum_{\substack{(i,j):i\in [K],j\in[2^K]}}{\hspace{-5mm}{c}_{i,j}I_{{\mathcal M}_{i,j}}(m)}
-\hspace{-4mm}\sum_{\substack{(i,j):(i,j)\in{\mathcal J}_c}}{\hspace{-5mm}{\beta}_{i,j}}
\\&
\overset{(c)}{=}\hspace{-5mm}\sum_{\substack{(i,j):i\in [K],j\in[2^K]}}{\hspace{-5mm}{c}_{i,j}I_{{\mathcal M}_{i,j}}(m)}-(e_m-d_m)
\\&=d_m=\min{(d_m,e_m)},
\end{aligned}
\end{equation}
\normalsize
where step $(a)$ follows from (\ref{def_Cija}). Step $(b)$ follows from the fact that $I_{{\mathcal M}_{i,j'}}(m)=0$ for $(i,j')\in {\mathcal N}_c$ and $I_{{\mathcal M}_{i,j}}(m)=1$ for $(i,j)\in {\mathcal J}_c$ according to (\ref{def_JNa}). Finally, step $(c)$ follows from (\ref{def_sum_c}).

Regarding ${\hat d}_{m'}$ and ${\hat e}_{m'}$ where $m'\neq m$, we have
\small
\begin{equation} \label{def_d_m_prim_c}
\begin{aligned}
{\hat d}_{m'}&=\sum_{j=1}^{2^K}{{\hat c}_{{m'},j}}\\&=\sum_{\substack{j:(m',j)\not\in {\mathcal J}_c, (m',j)\not\in {\mathcal N}_c}}{\hspace{-3mm}{\hat c}_{m',j}}
+\sum_{j:(m',j)\in {\mathcal N}_c}{{\hat c}_{m',j}}+
\sum_{\substack{j:(m',j)\in {\mathcal J}_c}}{\hspace{-3mm}{\hat c}_{m',j}}
\\&\overset{(a)}{=}\sum_{j=1}^{2^K}{{c}_{{m'},j}}
+\sum_{\substack{j':(m',j')\in {\mathcal N}_c}}{\hspace{-3mm}{\alpha}_{m',j'}}-\sum_{\substack{j:(m',j)\in {\mathcal J}_c}}{\hspace{-3mm}{\beta}_{m',j}}
\\&\overset{(b)}{=}\sum_{j=1}^{2^K}{{c}_{{m'},j}}
+\sum_{\substack{j:(m',j)\in {\mathcal J}_c}}{\hspace{-3mm}{\beta}_{m',j}}-\sum_{\substack{j:(m',j)\in {\mathcal J}_c}}{\hspace{-3mm}{\beta}_{m',j}}\\&=d_{m'},
\end{aligned}
\end{equation}
\normalsize
where step $(a)$ and $(b)$ follow from (\ref{def_Cija}) and (\ref{def_alpha_beta_c}), respectively.
\small
\begin{equation} \label{def_e_m_prim_c}
\begin{aligned}
{\hat e}_{m'}&=\sum_{i=1}^K{\sum_{j=1}^{2^K}{{\hat c}_{i,j}I_{{\mathcal M}_{i,j}}(m')}}\\&=\sum_{\substack{(i,j):(i,j)\not\in{\mathcal N}_c, (i,j)\not\in{\mathcal J}_c}}{\hspace{-3mm}{\hat c}_{i,j}I_{{\mathcal M}_{i,j}}(m')}
+\sum_{\substack{(i,j):(i,j)\in{\mathcal N}_c}}{\hspace{-3mm}{\hat c}_{i,j}I_{{\mathcal M}_{i,j}}(m')}+
\sum_{\substack{(i,j):(i,j)\in{\mathcal J}_c}}{\hspace{-3mm}{\hat c}_{i,j}I_{{\mathcal M}_{i,j}}(m')}
\\&\overset{(a)}{=}\sum_{\substack{(i,j):i\in [K],j\in[2^K]}}{\hspace{-3mm}{c}_{i,j}I_{{\mathcal M}_{i,j}}(m')}+\sum_{\substack{(i,j'):(i,j')\in{\mathcal N}_c}}{\hspace{-5mm}{\alpha}_{i,j'}I_{{\mathcal M}_{i,j'}}(m')}
-\sum_{\substack{(i,j):(i,j)\in{\mathcal J}_c}}{\hspace{-3mm}{\beta}_{i,j}I_{{\mathcal M}_{i,j}}(m')}
\\&\overset{(b)}{=} \sum_{\substack{(i,j):i\in [K],j\in[2^K]}}{\hspace{-3mm}{c}_{i,j}I_{{\mathcal M}_{i,j}}(m')}
+\sum_{\substack{(i,j):(i,j)\in{\mathcal J}_c}}{\hspace{-3mm}{\beta}_{i,j}I_{{\mathcal M}_{i,j}}(m')}-\sum_{\substack{(i,j):(i,j)\in{\mathcal J}_c}}{\hspace{-3mm}{\beta}_{i,j}I_{{\mathcal M}_{i,j}}(m')}
\\&=e_{m'},
\end{aligned}
\end{equation}
\normalsize
where step $(a)$ follows from (\ref{def_Cija}). Step $(b)$ follows from (\ref{def_alpha_beta_c}) and the fact $I_{{\mathcal M}_{i,j'}}(m')=I_{{\mathcal M}_{i,j}}(m')$ for $(i,j')\in {\mathcal N}_c$, $(i,j)\in {\mathcal J}_c$, and $m'\neq m$.
\end{proof}
Based on Claim 2 and Lemma 3, the projection of (\ref{ineq_a}) would result in a bound on $\sum_{k=1}^K{\min({\hat d}_k,{\hat e}_k)R_{k}}$ where $\min({\hat d}_k,{\hat e}_k)=\min({d}_k,{e}_k)$, $\forall k \in [K]$.

We now provide the following Claim to prove (\ref{condition_small_ineq_a_B}).
\begin{claim} \label{rhs_claim_c}
By considering (\ref{def_JNa})-(\ref{def_Cija}), we have (\ref{condition_small_ineq_a_B}).
\end{claim}
\begin{proof}
The proof is as follows
\small
\begin{equation} \label{rhs_c}
\begin{aligned}
\sum_{i=1}^K{\sum_{j=1}^{2^K}{{\hat c}_{i,j}H(Y_i|V_{{\mathcal M}_{i,j}^c})}}&\overset{(a)}{=}\sum_{\substack{(i,j):(i,j)\not\in{\mathcal N}_c, (i,j)\not\in{\mathcal J}_c}}{\hspace{-3mm}{c}_{i,j}H(Y_i|V_{{\mathcal M}_{i,j}^c})}
+\sum_{\substack{(i,j'):(i,j')\in{\mathcal N}_c}}{\hspace{-3mm}{\alpha}_{i,j'}H(Y_i|V_{{\mathcal M}_{i,j'}^c})}-
\sum_{\substack{(i,j):(i,j)\in{\mathcal J}_c}}{\hspace{-3mm}{\beta}_{i,j}H(Y_i|V_{{\mathcal M}_{i,j}^c})}
\\&\overset{(b)}{=}\hspace{-3mm}\sum_{\substack{(i,j):i\in [K],j\in[2^K]}}{\hspace{-6mm}{c}_{i,j}H(Y_i|V_{{\mathcal M}_{i,j}^c})}+\hspace{-3mm}\sum_{\substack{(i,j):(i,j)\in{\mathcal J}_c}}{\hspace{-3mm}{\beta}_{i,j}H(Y_i|V_mV_{{\mathcal M}_{i,j}^c})}
-\sum_{\substack{(i,j):(i,j)\in{\mathcal J}_c}}{\hspace{-3mm}{\beta}_{i,j}H(Y_i|V_{{\mathcal M}_{i,j}^c})}
\\&\overset{(c)}{\leq} \sum_{\substack{(i,j):i\in [K],j\in[2^K]}}{\hspace{-3mm}{c}_{i,j}H(Y_i|V_{{\mathcal M}_{i,j}^c})},
\end{aligned}
\end{equation}
\normalsize
where step $(a)$ follows from (\ref{def_Cija}). Step $(b)$ follows from (\ref{def_alpha_beta_c}) and the fact ${\mathcal M}_{i,j'}={\mathcal M}_{i,j}\backslash m$ for $(i,j')\in {\mathcal N}_c$, $(i,j)\in {\mathcal J}_c$. Finally, step $(c)$ follows from the fact $H(Y_i|V_mV_{{\mathcal M}_{i,j}^c})\leq H(Y_i|V_{{\mathcal M}_{i,j}^c})$ for all $(i,j)\in {\mathcal J}_c$.
\end{proof}

Therefore, it is obvious that if Claims 1-3 are satisfied, the bound obtained by projecting (\ref{ineq_a}) becomes tighter than the bound found by projecting (\ref{ref_non_eq_coef}).

By repeating the aforementioned process for all $m$ where $e_m > d_m$
and updating the resulting inequality, i.e. replacing coefficients $c_{i,j}$'s with ${\hat c}_{i,j}$'s,
we find inequality (\ref{ach_1}) which satisfies $a)$ (\ref{condition_big_ineq_a}) and $b)$ its projection leads to a tighter bound compared to the projection of (\ref{ref_non_eq_coef}).



\noindent{\bf Step 2:}
In this step, we aim to find new inequality
\small
\begin{equation} \label{ach_2}
\begin{aligned}
&\sum_{i=1}^K{\sum_{j=1}^{2^K}{{\bar c}_{i,j}(R_{ip}+\sum_{k \in {\mathcal M}_{i,j}}{R_{kc}})}}~~\leq \sum_{i=1}^K{\sum_{j=1}^{2^K}{{\bar c}_{i,j}H(Y_i|V_{{\mathcal M}_{i,j}^c})}},
\end{aligned}
\end{equation}
\normalsize
from (\ref{ach_1}) such that $a)$ ${\bar d}_{m}={\bar e}_{m}$, $\forall m \in [K]$ where ${\bar d}_{m}$ and ${\bar e}_{m}$ are similarly defined as (\ref{eq2conditions2}) and $b)$ the projection of (\ref{ach_2}), according to transformation matrix ${\bf A}$, results in a tighter bound on $\sum_{k=1}^K{\min(d_k,e_k)R_{k}}$ compared to the projection of (\ref{ach_1}).

Let us consider a $m$ where $d_m>e_m$ and define
\small
\begin{equation} \label{def_JNb}
\begin{aligned}
{\mathcal J}_p &\triangleq \{(i,j)|i= m,~m\not\in {\mathcal M}_{i,j}\},\\
{\mathcal N}_p^+ &\triangleq \{(i,j)|i\neq m,~i \in {\mathcal M}_{i,j}\},\\
{\mathcal N}_p^- &\triangleq \{(i,j)|i\neq m,~i \not \in {\mathcal M}_{i,j}\},\\
\end{aligned}
\end{equation}
\normalsize
\small
\begin{equation} \label{def_gamma_p}
\begin{aligned}
~~~~~~~~~~~{\gamma}_{m'} \triangleq \sum_{(k,j)\in {\mathcal J}_p}{{\alpha}_{k,j}I_{{\mathcal M}_{k,j}}(m')},~~\forall m'\neq m,
\end{aligned}
\end{equation}
\normalsize
\small
\begin{equation} \label{def_mu_p}
\begin{aligned}
~~~~~~~~~~~~~~~~~~~~~~~~~~~~~~~~~~~~~~{\mu}_{i,j'} \triangleq {\beta}_{i,j}~~\text{for $(i,j')\in {\mathcal N}_p^+$, $(i,j) \in {\mathcal N}_p^{-}$ s.t. ${\mathcal M}_{i,j'}={\mathcal M}_{i,j}\cup \{i\}$},
\end{aligned}
\end{equation}
\normalsize
for all ${\alpha}_{i,j},{\beta}_{i,j} \in \mathbb{Z}^{\geq 0}$ satisfying ${\alpha}_{i,j},{\beta}_{i,j}\leq {c}_{i,j}$ and
\small
\begin{equation} \label{def_sum_alpha_p}
\begin{aligned}
~~~~~~~\sum_{(i,j)\in {\mathcal J}_p} {{\alpha}_{i,j}}&=d_m-e_m~~~~~~~,~~~~~~~~{\alpha}_{i,j}&=0~~~~\forall (i,j)\not\in {\mathcal J}_p.
\end{aligned}
\end{equation}
\normalsize
\small
\begin{equation} \label{def_sum_beta_p}
\begin{aligned}
\sum_{j:(m',j)\in {\mathcal N}_p^-} {\hspace{-3mm}{\beta}_{m',j}}&={\gamma}_{m'}~~\forall m'\neq m~~~~~~~,~~~~~~~~~{\beta}_{i,j}&=0~~~~\forall (i,j)\not\in {\mathcal N}_p^-.
\end{aligned}
\end{equation}
\normalsize
By introducing coefficient ${\hat c}_{i,j}$ as follows
\small
\begin{equation} \label{def_Cijb}
\begin{aligned}
{\hat c}_{i,j}\triangleq
\left\{
	\begin{array}{ll}
		c_{i,j}  & \mbox{if } (i,j)\not\in {\mathcal J}_p~,~(i,j)\not \in {\mathcal N}_p^{+},~(i,j)\not \in {\mathcal N}_p^{-}, \\
		c_{i,j}+{\mu}_{i,j} & \mbox{if } (i,j)\in {\mathcal N}_p^{+},\\
        c_{i,j}-{\beta}_{i,j} & \mbox{if } (i,j)\in {\mathcal N}_p^{-},\\
		c_{i,j}-{\alpha}_{i,j}  & \mbox{if } (i,j)\in {\mathcal J}_p,
	\end{array}
\right.
\end{aligned}
\end{equation}
\normalsize
where $c_{i,j}\triangleq{\tilde c}_{i,j}$ for all $i,j$, we now consider the following inequality in region ${\mathcal A}_2$
\small
\begin{equation} \label{ineq_b}
\begin{aligned}
\sum_{i=1}^K{\sum_{j=1}^{2^K}{{\hat c}_{i,j}(R_{ip}+\sum_{k \in {\mathcal M}_{i,j}}{R_{kc}})}}~~\leq \sum_{i=1}^K{\sum_{j=1}^{2^K}{{\hat c}_{i,j}H(Y_i|V_{{\mathcal M}_{i,j}^c})}},
\end{aligned}
\end{equation}
\normalsize
and show that
\small
\begin{equation} \label{condition_small_ineq_b}
\begin{aligned}
(C)\quad
\left\{
\begin{array}{ll}
{\hat d}_{m}={\hat e}_{m}= e_{m}&\\{\hat d}_{m'}={d}_{m'}~&,~{\hat e}_{m'}= e_{m'}~~~~\forall m'\neq m,
\end{array}
\right.
\end{aligned}
\end{equation}
\normalsize
\small
\begin{equation} \label{condition_small_ineq_b_B}
\begin{aligned}
~~~~~~~~(D)\quad
\left\{
	\begin{array}{ll}
\sum \limits_{\substack{(i,j):\\i\in[K],j\in[2^K]}}{\hspace{-4mm}{\hat c}_{i,j}H(Y_i|V_{{\mathcal M}_{i,j}^c})} \leq \hspace{-4mm}{\sum \limits_{\substack{(i,j):\\i\in[K],j\in[2^K]}}{\hspace{-4mm}{c}_{i,j}H(Y_i|V_{{\mathcal M}_{i,j}^c})}}.
\end{array}
\right.
\end{aligned}
\end{equation}
\normalsize

We first present Claim \ref{exist_claim_p} to prove the existence of such $\alpha_{i,j}$'s and $\beta_{i,j}$'s satisfying (\ref{def_sum_alpha_p}) and (\ref{def_sum_beta_p}), respectively.
\begin{claim}\label{exist_claim_p}
By considering (\ref{def_JNb})-(\ref{def_gamma_p}), there exists ${\alpha}_{i,j}\leq {c}_{i,j}$ for all $(i,j)\in {\mathcal J}_p$ and ${\beta}_{i,j}\leq {c}_{i,j}$ for all $(i,j)\in {\mathcal N}_p^-$ satisfying (\ref{def_sum_alpha_p}) and (\ref{def_sum_beta_p}), respectively.
\end{claim}
\begin{proof}
One can easily prove that there exists such ${\alpha}_{i,j}\leq {c}_{i,j}$ and ${\beta}_{i,j}\leq {c}_{i,j}$ satisfying (\ref{def_sum_alpha_p}) and (\ref{def_sum_beta_p}) by showing that $\sum_{(i,j)\in {\mathcal J}_p} {{\alpha}_{i,j}}$ and $\sum_{(m',j)\in {\mathcal N}_p^-} {{\beta}_{m',j}}$ can respectively take values $d_m-e_m$ and $\gamma_{m'}$ for $m'\neq m$ as follows
\small
\begin{equation} \label{lem_exist_p}
\begin{aligned}
~~~~~~~~~~~~~~~~~~~~~~~~~~~~~~\sum_{(i,j)\in {\mathcal J}_p} {{\alpha}_{i,j}}&=\sum_{j:m \not\in {\mathcal M}_{m,j}} {{\alpha}_{m,j}} \\&\leq \sum_{j:m \not\in {\mathcal M}_{m,j}} {{c}_{m,j}}
\\& = \sum_{j:m \not\in {\mathcal M}_{m,j}} {{c}_{m,j}}
+\sum_{(i,j):m \in {\mathcal M}_{i,j}} {{c}_{i,j}}-\underbrace{\sum_{(i,j):m \in {\mathcal M}_{i,j}} {{c}_{i,j}}}_{=e_m}
\\&=\underbrace{\sum_{j:m \not\in {\mathcal M}_{m,j}} {{c}_{m,j}} + \sum_{j:m \in {\mathcal M}_{m,j}} {{c}_{m,j}}}_{=d_m}+
\underbrace{\sum_{\substack{(i,j):i\neq m, m \in {\mathcal M}_{i,j}}} {{c}_{i,j}}}_{\geq 0}-e_m,
\end{aligned}
\end{equation}
\normalsize
\small
\begin{equation} \label{lem_exist_p2}
\begin{aligned}
&\sum_{(m',j)\in {\mathcal N}_p^-} {{\beta}_{m',j}}=\sum_{j:m'\not\in {\mathcal M}_{m',j}} {{\beta}_{m',j}} \leq \sum_{j:m'\not\in {\mathcal M}_{m',j}} {{c}_{m',j}}.
\end{aligned}
\end{equation}
\normalsize
Since we have
\small
\begin{equation} \label{lem_exist_p3}
\begin{aligned}
&\hspace{-3mm}\sum_{\substack{j:\\m'\not\in {\mathcal M}_{m',j}}} {\hspace{-5mm}{c}_{m',j}} \overset{(a)}{\geq} \hspace{-3mm}\sum_{\substack{(i,j):i\neq m'\\, m'\in {\mathcal M}_{i,j}}} {\hspace{-4mm}{c}_{i,j}} \geq
\hspace{-3mm}\sum_{\substack{j:\\m'\in {\mathcal M}_{m,j}}} {\hspace{-4mm}{c}_{m,j}} \geq
\hspace{-5mm}\sum_{\substack{j:m \not\in {\mathcal M}_{m,j},\\m' \in {\mathcal M}_{m,j}}} {\hspace{-5mm}{c}_{m,j}} \geq
\hspace{-4mm}\sum_{\substack{j:m \not\in {\mathcal M}_{m,j},\\m' \in {\mathcal M}_{m,j}}} {\hspace{-5mm}{\alpha}_{m,j}}\overset{(b)}{=}\gamma_{m'},
\end{aligned}
\end{equation}
\normalsize
$\forall m'\neq m$ where step $(a)$ follows from $d_{m'}\geq e_{m'}$ implying
\small
\begin{equation} \label{lem_exist_p31}
\begin{aligned}
\sum_{j:m'\not\in {\mathcal M}_{m',j}} {\hspace{-5mm}{c}_{m',j}}+\hspace{-5mm}\sum_{j:m'\in {\mathcal M}_{m',j}} {\hspace{-5mm}{c}_{m',j}} \geq
\hspace{-5mm}\sum_{j:m'\in {\mathcal M}_{m',j}} {\hspace{-3mm}{c}_{m',j}}+\hspace{-5mm}\sum_{(i,j):i\neq m',m'\in {\mathcal M}_{i,j}} {\hspace{-3mm}{c}_{i,j}}.
\end{aligned}
\end{equation}
\normalsize
Step $(b)$ follows from (\ref{def_gamma_p}). Based on (\ref{lem_exist_p}), (\ref{lem_exist_p2}), and (\ref{lem_exist_p3}), the proof of existence of such ${\alpha}_{i,j}\leq {c}_{i,j}$ and ${\beta}_{i,j}\leq {c}_{i,j}$ completes.
\end{proof}

We next demonstrate Claim \ref{coeff_claim_p} to prove (\ref{condition_small_ineq_b}).
\begin{claim}\label{coeff_claim_p}
By considering (\ref{def_JNb})-(\ref{def_Cijb}), we have (\ref{condition_small_ineq_b}).
\end{claim}
\begin{proof}
We first find ${\hat d}_{m}$ and ${\hat e}_{m}$, then we derive ${\hat d}_{m'}$ and ${\hat e}_{m'}$ for all $m'\neq m$ as follows
\small
\begin{equation} \label{def_d_m_p}
\begin{aligned}
{\hat d}_m&=\sum_{j=1}^{2^K}{{\hat c}_{m,j}}\\&\overset{(a)}{=}\sum_{j=1}^{2^K}{c_{m,j}}+\hspace{-5mm}\sum_{j:(m,j)\in {\mathcal N}_p^+ }{\hspace{-4mm}{\mu}_{m,j}}-\hspace{-5mm}\sum_{j:(m,j)\in {\mathcal N}_p^- }{\hspace{-4mm}{\beta}_{m,j}}
-\hspace{-5mm}\sum_{j:(m,j)\in {\mathcal J}_p }{\hspace{-4mm}{\alpha}_{m,j}}
\\&\overset{(b)}{=}\sum_{j=1}^{2^K}{c_{m,j}}- \sum_{j:(m,j)\in {\mathcal J}_p }{\hspace{-4mm}{\alpha}_{m,j}}
\\&\overset{(c)}{=}d_m-(d_m-e_m)=e_m,
\end{aligned}
\end{equation}
\normalsize
where step $(a)$ follows from (\ref{def_Cijb}). Step $(b)$ follows from the fact that $(m,j)\not\in {\mathcal N}_p^+$ and $(m,j)\not\in {\mathcal N}_p^-$. Finally, step $(c)$ follows from (\ref{def_sum_alpha_p}).

\small
\begin{equation} \label{def_e_m_p}
\begin{aligned}
{\hat e}_{m}&=\sum_{i=1}^K{\sum_{j=1}^{2^K}{{\hat c}_{i,j}I_{{\mathcal M}_{i,j}}(m)}}
\\&\overset{(a)}{=}\sum_{\substack{(i,j):i\in [K],j\in[2^K]}}{\hspace{-3mm}{c}_{i,j}I_{{\mathcal M}_{i,j}}(m)}
+\sum_{\substack{(i,j'):(i,j')\in{\mathcal N}_p^+}}{\hspace{-3mm}{\mu}_{i,j'}I_{{\mathcal M}_{i,j'}}(m)}-
\sum_{\substack{(i,j):(i,j)\in{\mathcal N}_p^-}}{\hspace{-3mm}{\beta}_{i,j}I_{{\mathcal M}_{i,j}}(m)}
-\sum_{\substack{(i,j):(i,j)\in{\mathcal J}_p}}{\hspace{-3mm}{\alpha}_{i,j}\underbrace{I_{{\mathcal M}_{i,j}}(m)}_{\overset{(b)}{=}0}}
\\&\overset{(c)}{=}\sum_{\substack{(i,j):i\in [K],j\in[2^K]}}{\hspace{-3mm}{c}_{i,j}I_{{\mathcal M}_{i,j}}(m)}
+\sum_{\substack{(i,j):(i,j)\in{\mathcal N}_p^-}}{\hspace{-3mm}{\beta}_{i,j}I_{{\mathcal M}_{i,j}}(m)}-\sum_{\substack{(i,j):(i,j)\in{\mathcal N}_p^-}}{\hspace{-3mm}{\beta}_{i,j}I_{{\mathcal M}_{i,j}}(m)}\\&=e_m,
\end{aligned}
\end{equation}
\normalsize
where step $(a)$ follows from (\ref{def_Cijb}) and step $(b)$ follows from $m\not\in {\mathcal M}_{i,j}$ for $(i,j)\in{\mathcal J}_p$. Finally, step $(c)$ follows from (\ref{def_mu_p}), meaning that $\mu_{i,j'}={\beta}_{i,j}$ and $I_{{\mathcal M}_{i,j'}}(m)=I_{{\mathcal M}_{i,j}}(m)$ for $(i,j)\in {\mathcal N}_p^{-}$ and $(i,j')\in {\mathcal N}_p^+$.

Regarding ${\hat d}_{m'}$ and ${\hat e}_{m'}$ where $m'\neq m$, we have
\small
\begin{equation} \label{def_d_m_prim_p}
\begin{aligned}
{\hat d}_{m'}&=\sum_{j=1}^{2^K}{{\hat c}_{{m'},j}}
\\&\overset{(a)}{=}\sum_{j=1}^{2^K}{{c}_{{m'},j}}
+\hspace{-3mm}\sum_{\substack{j':(m',j')\in {\mathcal N}_p^+}}{\hspace{-3mm}{\mu}_{m',j'}}-\hspace{-3mm}\sum_{\substack{j:(m',j)\in {\mathcal N}_p^-}}{\hspace{-3mm}{\beta}_{m',j}}
-\underbrace{\sum_{\substack{j:(m',j)\in {\mathcal J}_p}}{\hspace{-3mm}{\alpha}_{m',j}}}_{\overset{(b)}{=}0} \\&\overset{(c)}{=}\sum_{j=1}^{2^K}{{c}_{{m'},j}}
+\hspace{-5mm}\sum_{\substack{j:(m',j)\in {\mathcal N}_p^-}}{\hspace{-3mm}{\beta}_{m',j}}-\hspace{-5mm}\sum_{\substack{j:(m',j)\in {\mathcal N}_p^-}}{\hspace{-3mm}{\beta}_{m',j}}
\\&=d_{m'},
\end{aligned}
\end{equation}
\normalsize
where steps $(a)$, $(b)$, and $(c)$ follow from (\ref{def_Cijb}), $(m',j)\not\in{\mathcal J}_{p}$, and (\ref{def_mu_p}), respectively.

\small
\begin{equation} \label{def_e_m_prim_p}
\begin{aligned}
{\hat e}_{m'}&=\sum_{i=1}^K{\sum_{j=1}^{2^K}{{\hat c}_{i,j}I_{{\mathcal M}_{i,j}}(m')}}
\\&\overset{(a)}{=}\sum_{\substack{(i,j):i\in [K],j\in[2^K]}}{\hspace{-3mm}{c}_{i,j}I_{{\mathcal M}_{i,j}}(m')}
\underbrace{-\hspace{-5mm}\sum_{\substack{(i,j):(i,j)\in{\mathcal J}_p}}{\hspace{-6mm}{\alpha}_{i,j}I_{{\mathcal M}_{i,j}}(m')}}_{\overset{(b)}{=}-{\gamma}_{m'}}-\hspace{-2mm}\sum_{\substack{(i,j):\\(i,j)\in{\mathcal N}_p^-}}{\hspace{-3mm}{\beta}_{i,j}I_{{\mathcal M}_{i,j}}(m')}
+\hspace{-2mm}\sum_{\substack{(i,j):\\(i,j)\in{\mathcal N}_p^+}}{\hspace{-3mm}{\mu}_{i,j}I_{{\mathcal M}_{i,j}}(m')}
\\&= e_{m'}-{\gamma}_{m'}-\hspace{-6mm}\sum_{\substack{j:(m',j)\in{\mathcal N}_p^-}}{\hspace{-3mm}{\beta}_{m',j}\underbrace{I_{{\mathcal M}_{m',j}}(m')}_{\overset{(c)}{=}0}}-\hspace{-6mm}\sum_{\substack{(i,j)\in{\mathcal N}_p^-,i\neq m'}}{\hspace{-3mm}{\beta}_{i,j}I_{{\mathcal M}_{i,j}}(m')}
+\hspace{-4mm}\sum_{\substack{j:(m',j)\in{\mathcal N}_p^+}}{\hspace{-4mm}{\mu}_{m',j}\underbrace{I_{{\mathcal M}_{m',j}}(m')}_{\overset{(d)}{=}1}}+\hspace{-4mm}\sum_{\substack{(i,j')\in{\mathcal N}_p^+,i\neq m'}}{\hspace{-6mm}{\mu}_{i,j'}I_{{\mathcal M}_{i,j'}}(m')}
\\&\overset{(e)}{=}e_m-{\gamma}_{m'}\nonumber
-\hspace{-4mm}\sum_{\substack{(i,j)\in{\mathcal N}_p^-,i\neq m'}}{\hspace{-6mm}{\beta}_{i,j}I_{{\mathcal M}_{i,j}}(m')}+
\hspace{-4mm}\sum_{\substack{j:(m',j)\in{\mathcal N}_p^+}}{\hspace{-3mm}{\mu}_{m',j}}+\hspace{-6mm}\sum_{\substack{(i,j)\in{\mathcal N}_p^-,i\neq m'}}{\hspace{-3mm}{\beta}_{i,j}I_{{\mathcal M}_{i,j}}(m')}
\\&
\overset{(f)}{=}e_m-{\gamma}_{m'}+{\gamma}_{m'}=e_{m'},
\end{aligned}
\end{equation}
\normalsize
where steps $(a)$ and $(b)$ follow from (\ref{def_Cijb}) and (\ref{def_gamma_p}), respectively. Steps $(c)$ and $(d)$ follow from (\ref{def_JNb}). Finally, steps $(e)$ and $(f)$ follow from (\ref{def_mu_p}) and $(\ref{def_sum_beta_p})$.
\end{proof}
Based on Claim 5 and Lemma 3, the projection of (\ref{ineq_b}) would result in a bound on $\sum_{k=1}^K{\min({\hat d}_k,{\hat e}_k)R_{k}}$ where $\min({\hat d}_k,{\hat e}_k)=\min({d}_k,{e}_k)$, $\forall k \in [K]$.

Finally, we prove (\ref{condition_small_ineq_b_B}) in the following Claim.
\begin{claim} \label{rhs_claim_p}
By considering (\ref{def_JNb})-(\ref{def_Cijb}), we have (\ref{condition_small_ineq_b_B}).
\end{claim}
\begin{proof}
\small
\begin{equation} \label{rhs_p}
\begin{aligned}
\sum_{i=1}^K{\sum_{j=1}^{2^K}{{\hat c}_{i,j}H(Y_i|V_{{\mathcal M}_{i,j}^c})}}&\overset{(a)}{=}\sum_{\substack{(i,j):i\in [K],j\in[2^K]}}{{c}_{i,j}H(Y_i|V_{{\mathcal M}_{i,j}^c})}
+\hspace{-3mm}\sum_{\substack{(i,j'):(i,j')\in{\mathcal N}_p^+}}{\hspace{-3mm}{\mu}_{i,j'}H(Y_i|V_{{\mathcal M}_{i,j'}^c})}-
\hspace{-3mm}\sum_{\substack{(i,j):(i,j)\in{\mathcal N}_p^-}}{{\beta}_{i,j}H(Y_i|V_{{\mathcal M}_{i,j}^c})}
\\&~~~~~-\sum_{\substack{(i,j):(i,j)\in{\mathcal J}_p}}{{\alpha}_{i,j}H(Y_i|V_{{\mathcal M}_{i,j}^c})}\\&\overset{(b)}{=}\hspace{-6mm}\sum_{\substack{(i,j):i\in [K],j\in[2^K]}}{\hspace{-2mm}{c}_{i,j}H(Y_i|V_{{\mathcal M}_{i,j}^c})}
+\sum_{\substack{(i,j'):(i,j')\in{\mathcal N}_p^+}}{\underbrace{{\mu}_{i,j'}H(Y_i|V_{{\mathcal M}_{i,j'}^c})-{\mu}_{i,j'}H(Y_i|V_iV_{{\mathcal M}_{i,j'}^c})}_{\overset{(c)}{\leq}{\mu}_{i,j'}H(V_i)}}
\\&~~~~~-\sum_{\substack{(i,j):(i,j)\in{\mathcal J}_p}}{{\alpha}_{i,j}H(Y_i|V_{{\mathcal M}_{i,j}^c})}
\\&\overset{(d)}{\leq} \hspace{-6mm}\sum_{\substack{(i,j):i\in [K],j\in[2^K]}}{\hspace{-3mm}{c}_{i,j}H(Y_i|V_{{\mathcal M}_{i,j}^c})}
+\hspace{-4mm}\sum_{\substack{(i,j'):(i,j')\in{\mathcal N}_p^+}}{\hspace{-4mm}{\mu}_{i,j'}H(V_i)}-\hspace{-4mm}\sum_{\substack{j:(m,j)\in{\mathcal J}_p}}{\hspace{-4mm}{\alpha}_{m,j}H(Y_m|
V_mV_{\{{\mathcal M}_{i,j}\cup \{m\}\}^c})}
\end{aligned}
\end{equation}
\normalsize
where steps $(a)$ and $(b)$ follow from (\ref{def_Cijb}) and (\ref{def_mu_p}), respectively. Step $(c)$ follows from the independence of $V_i$'s and the fact that $I(Y_i;V_i|V_{\mathcal B})\leq H(V_i|V_{\mathcal B})=H(V_i)$ for $i\not\in {\mathcal B}$, which implies $H(Y_i|V_{{\mathcal M}_{i,j'}^c})-H(Y_i|V_iV_{{\mathcal M}_{i,j'}^c})\leq H(V_i)$ where $i \not\in {\mathcal M}_{i,j'}^c$. Furthermore, step $(d)$ follows from the fact that ${\alpha}_{i,j}=0$ for $\forall i\neq m$ and $m \not\in {\mathcal M}_{i,j}$ for $(i,j)\in {\mathcal J}_p$.

We can find an upper bound on (\ref{rhs_p}) by using $H(Y_m|V_mV_{\{{\mathcal M}_{i,j}\cup \{m\}\}^c})\geq H(Y_m|X_mV_{\{{\mathcal M}_{i,j}\cup \{m\}\}^c})=H(V_{{\mathcal M}_{i,j}})$ due to $V_m=g_m(X_m)$, $m \not\in {\mathcal M}_{i,j}$ for $(i,j)\in {\mathcal J}_p$, and (\ref{eq02a}) as follows
\small
\begin{equation} \label{rhs_p_rewrite}
\begin{aligned}
\sum_{i=1}^K{\sum_{j=1}^{2^K}{{\hat c}_{i,j}H(Y_i|V_{{\mathcal M}_{i,j}^c})}}
&\leq \sum_{\substack{(i,j):i\in [K],j\in[2^K]}}{{c}_{i,j}H(Y_i|V_{{\mathcal M}_{i,j}^c})}
+\sum_{\substack{(i,j'):(i,j')\in{\mathcal N}_p^+}}{{\mu}_{i,j'}H(V_i)}-\sum_{\substack{j:(m,j)\in{\mathcal J}_p}}{{\alpha}_{m,j}H(V_{{\mathcal M}_{m,j}})}
\\&\overset{(a)}{=}\hspace{-6mm}\sum_{\substack{(i,j):i\in [K],j\in[2^K]}}{\hspace{-6mm}{c}_{i,j}H(Y_i|V_{{\mathcal M}_{i,j}^c})}+\hspace{-5mm}\sum_{\substack{(i,j'):(i,j')\in{\mathcal N}_p^+}}{\hspace{-5mm}{\mu}_{i,j'}H(V_i)}-\hspace{-2mm}\sum_{m'\neq m}{\hspace{-2mm}{\gamma}_{m'}H(V_{m'})}
\\&{=}\hspace{-3mm}\sum_{\substack{(i,j):i\in [K],\\j\in[2^K]}}{\hspace{-3mm}{c}_{i,j}H(Y_i|V_{{\mathcal M}_{i,j}^c})}
+\sum_{m'\neq m}{H(V_{m'})\Big(\underbrace{\sum_{j':(m',j')\in {\mathcal N}_p^+}{{\mu}_{m',j'}}}_{\overset{(b)}{=}\gamma_{m'}}\Big)}
-\hspace{-2mm}\sum_{m'\neq m}{\hspace{-2mm}{\gamma}_{m'}H(V_{m'})}\\&=\sum_{\substack{(i,j):i\in [K],j\in[2^K]}}{\hspace{-3mm}{c}_{i,j}H(Y_i|V_{{\mathcal M}_{i,j}^c})}
\end{aligned}
\end{equation}
\normalsize
where step $(a)$ follows from the fact that $H(V_{{\mathcal M}_{m,j}})=\sum_{m'\neq m}{I_{{\mathcal M}_{m,j}}(m')H(V_{m'})}$ for $(m,j)\in {\mathcal J}_{p}$ and (\ref{def_gamma_p}). Step $(b)$ follows from (\ref{def_mu_p}) and (\ref{def_sum_beta_p}).
\end{proof}

Therefore, it is easy to verify the projection of (\ref{ineq_b}) leads to a tighter bound compared to the projection of (\ref{ach_1}) if Claims 4-6 are satisfied.

By repeating the above process for all $m$ where $d_m>e_m$ and updating the resulting inequality,
i.e. replacing coefficients $c_{i,j}$'s with ${\hat c}_{i,j}$'s,
we find inequality (\ref{ach_2}) such that $a)$ ${\bar d}_{m}={\bar e}_{m}$, $\forall m \in [K]$ and $b)$ its projection leads to a tighter bound compared to the projection of (\ref{ref_non_eq_coef}). Note that the projection of (\ref{ach_2}) leads to a tighter bound compared to the projection of (\ref{ach_1}) and the projection of (\ref{ach_1}) results in a tighter bound compared to the projection of (\ref{ref_non_eq_coef}).
\end{proof}

We now show that region ${\mathcal A}_3$ is matching region ${\mathcal A}$ according to the following lemma.
\begin{lem}
\label{lem2}
Region ${\mathcal A}_3$ is the same as region ${\mathcal A}$, i.e.
\small
\begin{equation} \label{eq03a2}
\begin{aligned}
{\mathcal A}_3={\mathcal A}.
\end{aligned}
\end{equation}
\normalsize
\end{lem}

\begin{proof}[{Proof}]
In order to prove this, we need to show that we can express any inequalities of region ${\mathcal A}_3$ in terms of inequalities of region ${\mathcal A}$ and vice versa.
First, we consider an inequality $\sum_{i=1}^{K}{a_iR_i} \leq \sum_{i=1}^{K}{\sum_{j=1}^{a_i}{H(Y_i|V_{{\mathcal S}_{i,j}^c})}}$ in region ${\mathcal A}$. It is easy to see that this inequality can be found by setting following parameters of region ${\mathcal A}_3$
\begin{equation} \label{eq0141a}
\begin{aligned}
{\mathcal S}_{i,j}={\mathcal M}_{i,q}~~~\text{s.t.}~~ q=\text{arg}{\min}_{k}{ \sum_{l=1}^{k}{c_{i,l}} \geq j}.
\end{aligned}
\end{equation}

Similarly, consider an inequality $\sum_{i=1}^K{a_iR_i}\leq \sum_{i=1}^K{\sum_{j=1}^{2^K}{c_{i,j}H(Y_i|V_{{\mathcal M}_{i,j}^c})}}$ in achievable region ${\mathcal A}_3$. This inequality can be obtained by setting following parameters
\begin{equation} \label{eq011a}
\begin{aligned}
c_{i,j}=\sum_{q=1}^{a_i}{I({\mathcal S}_{i,q}={\mathcal M}_{i,j})},
\end{aligned}
\end{equation}
where
\begin{equation} \label{eq011b}
{I({\mathcal A}={\mathcal B})}\triangleq \begin{cases} 1 &\mbox{if } {\mathcal A}={\mathcal B}, \\
0 & \mbox{otherwise. }
\end{cases}
\end{equation}
\end{proof}

\section{Converse} \label{converse}
Consider a coding scheme with rate-tuple $(R_1,...,R_K)$ in a $K$-user DIC, in particular encoder function $\Phi_i^n: {[1:2^{nR_i}]}\rightarrow {\mathcal X}_{S_i}^{n}$ and decoder function $\Psi_i^n: {\mathcal Y}_{D_i}^{n} \rightarrow {[1:2^{nR_i}]}$ for $i=1,...,K$, with vanishing error probability for sufficiently large $n$. Our goal is to show that there exists a product distribution $\prod_{i=1}^K{p(x_i)}$ such that (\ref{eq03a}) holds for all choices of $a_i$'s and ${\mathcal S}_{i,j}$'s.

According to Fano's inequality, we have
\small
\begin{equation} \label{converse_diff1}
\begin{aligned}
n\sum_{i=1}^K{a_iR_i}&\leq \sum_{i=1}^K{a_iI({W}_i;{Y}_i^n)}+n{\epsilon}_n
\\&{\leq}\sum_{i=1}^K{a_iI({X}_i^n;{Y}_i^n)} +n{\epsilon}_n\\&= \sum_{i=1}^K{\big(a_iH({Y}_i^n)-a_iH({Y}_i^n|{X}_i^n)\big)}+n{\epsilon}_n
\\&\overset{(a)}{=}\sum_{i=1}^K{\Big(a_iH({Y}_i^n)-a_i(\sum_{j\neq i}{H({V}_j^n)})\Big)}+n{\epsilon}_n
\\&\overset{(b)}{=}\sum_{i=1}^K{\Big(a_iH({Y}_i^n)-(L-a_i)(H({V}_i^n))\Big)}+n{\epsilon}_n
\\&\overset{(c)}{=}\sum_{i=1}^K{\sum_{q=1}^{a_i}{\big(H({Y}_i^n)-H({V}_{{\mathcal \phi}_{i,q}}^n)\big)}}+n{\epsilon}_n,
\end{aligned}
\end{equation}
\normalsize
where $W_i$ represents the message of $i$th source and step $(a)$ follows from (2) and the independence of $V_i$'s. Step $(b)$ can be derived by letting $L\triangleq \sum_{i=1}^K{a_i}$. Finally, step $(c)$ follows from the independence of $V_i$'s and introducing all subsets ${\mathcal \phi}_{i,q} \subseteq [K]$ such that $\sum_{i=1}^K{\sum_{q=1}^{a_i}{I_{{\mathcal \phi}_{i,q}}(m)=L-a_m}}$, $\forall m\in [K]$.

By letting ${\mathcal S}_{i,q}=[K]\backslash {\mathcal \phi}_{i,q}$, we can simplify (\ref{converse_diff1}) as follows
\small
\begin{equation} \label{converse_diff2}
\begin{aligned}
n&(\sum_{i=1}^K{a_iR_i}-{\epsilon}_n)&\leq\sum_{i=1}^K{\sum_{q=1}^{a_i}{\big(H({Y}_i^n)-H({V}_{{\mathcal S}_{i,q}^c}^n)\big)}},
\end{aligned}
\end{equation}
\normalsize
for all subsets ${\mathcal S}_{i,q}$'s satisfying $\sum_{i=1}^K{\sum_{q=1}^{a_i}{I_{{\mathcal S}_{i,q}}(m)=a_m}}$ for $m=1,...,K$. Note that $\sum_{i=1}^K{\sum_{q=1}^{a_i}{I_{{\mathcal S}_{i,q}}(m)}}=\sum_{i=1}^K{\sum_{q=1}^{a_i}{\Big(1-I_{{\mathcal \phi}_{i,q}}(m)\Big)=a_m}}$, $\forall m\in[K]$.

Considering inequality $H(Y)-H(X)\leq H(Y|X)$, which is due to $H(Y)-H(Y|X)=I(X;Y)\leq H(X)$, (\ref{converse_diff2}) can be simplified more as
\small
\begin{equation} \label{converse_diff3}
\begin{aligned}
n&(\sum_{i=1}^K{a_iR_i}-{\epsilon}_n)\leq\sum_{i=1}^K{\sum_{q=1}^{a_i}{H({Y}_i^n|{V}_{{\mathcal S}_{i,q}^c}^n)}}
\leq \sum_{m=1}^n{\sum_{i=1}^K{\sum_{q=1}^{a_i}{H({Y}_{i,m}|{V}_{{\mathcal S}_{i,q}^c,m})}}},
\end{aligned}
\end{equation}
\normalsize
where $Y_i^n\triangleq(Y_{i,1},...,Y_{i,n})$, ${V}_{{\mathcal S}_{i,q}^c}^n\triangleq({V}_{{\mathcal S}_{i,q}^c,1},...,{V}_{{\mathcal S}_{i,q}^c,n})$, and
all subsets ${\mathcal S}_{i,q}$'s satisfying $\sum_{i=1}^K{\sum_{q=1}^{a_i}{I_{{\mathcal S}_{i,q}}(m)=a_m}}$ for $m=1,...,K$.

By letting $n\rightarrow \infty$ and considering (\ref{converse_diff3}) for all $a_i$'s and ${\mathcal S}_{i,q}$'s as well as the convexity of $\mathcal A$, we can conclude that there exists a product distribution satisfying (4) for all choices of $a_i$'s and ${\mathcal S}_{i,q}$'s.
\section{Conclusion} \label{conclusion}
We considered the symmetric injective $K$-user deterministic interference channel and characterized the corresponding capacity region. The achievable rate region was obtained by projecting the achievable rate region of Han-Kobayashi scheme along the the direction of sum of common and private rates for each user. Furthermore, we derived a tight converse for the achievable rate region.
An interesting future direction can be characterizing the capacity region of general $K$-user deterministic interference channel.
\bibliographystyle{IEEEtran}
{\footnotesize
\bibliography{main_longer_arxiv_v1}
\appendices
\section{}\label{error_com_priv}
In this part, we analyze the error probability of $K$-user DIC. Since the analysis is the same for Receiver $i$ for $i=2,...,K$, we only analyze the error probability at Receiver 1. Furthermore, due to symmetry in generating the codewords, the average error does not depend on which message is sent. Therefore, we can assume message indexed by $(c_i,p_i)=(1,1)$ is sent by Transmitter $i$.

An error occurs if either the wrong codewords of Transmitter 1 are jointly typical with the received sequence or the correct codeword is not jointly typical with the received sequence. Let us define following events
\small
\begin{equation} \label{eqaap1}
\begin{aligned}
\emph{E}_{c_1,p_1,c_2,...,c_K}=&\{(X_1^n(c_1,p_1),V_1^n(c_1),V_2^n(c_2),...,V_K^n(c_K),Y_1^n) \in \emph{A}_{\epsilon}^n\}.
\end{aligned}
\end{equation}
\normalsize

Therefore, the error probability would be
\small
\begin{equation} \label{prb1}
\begin{aligned}
P_e&=P(\emph{E}_{1,1,1,...,1}^c~\bigcup~\cup_{(c_1,p_1)\neq(1,1)}{\emph{E}_{c_1,p_1,c_2,...,c_K}})
\\&< P(\emph{E}_{1,1,1,...,1}^c)+\sum_{(c_1,p_1)\neq(1,1)}{P(\emph{E}_{c_1,p_1,c_2,...,c_K})}
\\&=P(\emph{E}_{1,1,1,...,1}^c)+\sum_{c_1\neq1,p_1\neq1}{P(\emph{E}_{c_1,p_1,c_2,...,c_K})}
+\sum_{c_1\neq1,p_1=1}{P(\emph{E}_{c_1,p_1,c_2,...,c_K})}+\sum_{c_1=1,p_1\neq1}{P(\emph{E}_{c_1,p_1,c_2,...,c_K})}.
\end{aligned}
\end{equation}
\normalsize
Let us define function $\eta : {\mathbb{N}}\rightarrow \{0,1\}$ as follows
\small
\begin{equation} \label{etaaaa}
\eta(x)=\begin{cases} 1 &\mbox{if } x =1 \\
0 & \mbox{otherwise. }
\end{cases}
\begin{aligned}
\end{aligned}
\end{equation}
\normalsize
By indexing all possible $(\eta(c_2),...,\eta(c_K))$ and considering the $m$th index, we define $\gamma_m\triangleq(c_2,...,c_K)$ for all $m \in[2^{K-1}]$, hence we rewrite (\ref{prb1}) as follows
\small
\begin{equation} \label{eqaap12}
\begin{aligned}
\\&p_e\leq P(\emph{E}_{1,1,1,...,1}^c)+\underbrace{\sum_{c_1\neq1,p_1\neq1}{\sum_{m=1}^{2^{K-1}}{P(\emph{E}_{c_1,p_1,\gamma_m})}}}_{A},
+\underbrace{\sum_{c_1\neq1,p_1=1}{\sum_{m=1}^{2^{K-1}}{P(\emph{E}_{c_1,p_1,\gamma_m})}}}_{B}+\underbrace{\sum_{c_1=1,p_1\neq1}{\sum_{m=1}^{2^{K-1}}{P(\emph{E}_{c_1,p_1,\gamma_m})}}}_{C}.
\end{aligned}
\end{equation}
\normalsize
As $n$ goes to infinity, the first term goes to zero. Each term of $A$ goes to zero if the following condition is hold
\small
\begin{equation} \label{eqA}
\begin{aligned}
R_{1p}+R_{1c}+\sum_{k\in \tau_{m}}{R_{kc}}&\leq I(X_1,V_1,V_{\tau_{m}};Y|V_{[2:K]\backslash \tau_m})\\&=H(Y_1|V_{[2:K]\backslash \tau_m}),
\end{aligned}
\end{equation}
\normalsize
for $m=1,...,2^{K-1}$ where $\tau_m=\{i|j_i\neq1~\text{for}~i=2,...,K\}$. Similarly, each term of $B$ and $C$ goes to zero if the following constraints are hold
\small
\begin{equation} \label{eqBC}
\begin{aligned}
R_{1c}+\sum_{k\in \tau_{m}}{R_{kc}}&\leq I(X_1,V_1,V_{\tau_{m}};Y|V_{[2:K]\backslash \tau_m})\\&=H(Y_1|V_{[2:K]\backslash \tau_m}),
\\\sum_{k\in \tau_{m}}{R_{kc}}&\leq I(X_1,V_{\tau_{m}};Y|V_1,V_{[2:K]\backslash \tau_m})\\&=H(Y_1|V_1V_{[2:K]\backslash \tau_m}),
\end{aligned}
\end{equation}
\normalsize
for $m=1,...,2^{K-1}$. By removing redundant bounds of (\ref{eqA}) and (\ref{eqBC}), we have
\small
\begin{equation} \label{eqBC2}
\begin{aligned}
R_{1p}+\sum_{k\in \tau}{R_{kc}}&\leq H(Y_1|V_{{\tau}^c}),~~\forall \tau \subseteq \{1,...,K\},
\\R_{kp},R_{kc}&\geq 0,~~\forall k\in[K].
\end{aligned}
\end{equation}
\normalsize

By considering the error probability of all receivers and indexing all subsets of $\{1,...,K\}$, we have
\small
\begin{equation} \label{eq-achiv-rc-rp}
\begin{aligned}
R_{ip}+\sum_{k\in {\mathcal M}_{i,j}}{R_{kc}}&\leq H(Y_i|V_{{\mathcal M}_{i,j}^c}),~\forall i\in[K],~\forall j\in[2^K],
\\R_{kp},R_{kc}&\geq 0,~~\forall k\in[K],
\end{aligned}
\end{equation}
\normalsize
where ${\mathcal M}_{i,j}\subseteq[K]$ represents the $j$th subset of the corresponding $i$th receiver.

It should be noted one can easily verify that (\ref{eq-achiv-rc-rp}) matches the region found in Appendix I of \cite{jaafaar} for the case of $K=3$.


\section{}\label{app_proj_min}
\begin{lem}
\label{proj_min}
If $\sum_{i=1}^K{\sum_{j=1}^{2^K}{{c}_{i,j}(R_{ip}+\sum_{k \in {\mathcal M}_{i,j}}{R_{kc}})}}~~\leq \sum_{i=1}^K{\sum_{j=1}^{2^K}{{c}_{i,j}H(Y_i|V_{{\mathcal M}_{i,j}^c})}}$, then the projection of this inequality according to transformation matrix $\bf A$ would result in
\small
\begin{equation} \label{min_ineq_lemma}
\begin{aligned}
\sum_{i=1}^K{\min(d_i,e_i)R_{i}}\leq \sum_{i=1}^K{\sum_{j=1}^{2^K}{{c}_{i,j}H(Y_i|V_{{\mathcal M}_{i,j}^c})}}.
\end{aligned}
\end{equation}
\normalsize

\end{lem}
\begin{proof}[{Proof}]
The proof is based on Fourier-Motzkin elimination by considering $\sum_{i=1}^K{d_iR_{ip}+e_iR_{ic}} \leq \sum_{i=1}^K{\sum_{j=1}^{2^K}{{c}_{i,j}H(Y_i|V_{{\mathcal M}_{i,j}^c})}}$ and inequalities $R_{ip}+R_{ic}\leq R_i$, $R_{ip}+R_{ic}\geq R_i$, $R_{ic}\geq 0~,~R_{ip}\geq 0$ for $\forall i\in [K]$. Without loss of generality, we assume that $d_j\geq e_j$ for a specific $j \in [K]$. We first eliminate $R_{jc}$ as follows
\small
\begin{equation} \label{proofmin_kochik}
\begin{aligned}
    \begin{array}{cc}
        R_j-R_{jp}\\
        0
    \end{array}
    \Bigr\}\leq R_{jc}\leq \Bigl\{
    \begin{array}{cc}
        R_j-R_{jp}\\
        {{1}\over{e_j}}(\theta-d_jR_{jp}-\sum_{i\neq j}{d_iR_{ip}+e_iR_{ic}})
    \end{array}
\end{aligned}
\end{equation}
\normalsize
where $\theta\triangleq \sum_{i=1}^K{\sum_{j=1}^{2^K}{{c}_{i,j}H(Y_i|V_{{\mathcal M}_{i,j}^c})}}$. We next eliminate $R_{jp}$ as follows
\small
\begin{equation} \label{proofmin_kochik}
\begin{aligned}
    0\leq R_{jp}&\leq \Bigl\{
    \begin{array}{cc}
        R_j\\
        {{1}\over{d_j}}(\theta-\sum_{i\neq j}{d_iR_{ip}+e_iR_{ic}})
    \end{array}
\\(d_j-e_j)R_{jp}&\leq \theta-e_jR_j-\sum_{i\neq j}{d_iR_{ip}+e_iR_{ic}}
\end{aligned}
\end{equation}
\normalsize
therefore, after eliminating $R_{jc}$ and $R_{jp}$, we obtain inequalities $\sum_{i\neq j}{d_iR_{ip}+e_iR_{ic}}\leq \theta-\min{(d_j,e_j)}R_j$ and $R_j\geq0$. Note that we first eliminate $R_{jp}$ then $R_{jc}$ in the case of $e_j\geq d_{j}$. By performing the similar technique to the resulting inequalities and eliminating the remaining common and private rates, we obtain (\ref{min_ineq_lemma}).
\end{proof}

\onecolumn
\section{}\label{Table_3user}

\begin{table}[ht]
\centering 
\begin{tabular}{|c|c|c|} 
\hline 
$a_i$'s & ${\mathcal S}_{i,j}$'s & Bound  \\ 
\hline 
$a_1=1$&${\mathcal S}_{1,1}=\{1\}$& (10) of \cite{jaafaar}\\ 
$a_1=1,a_2=1$&${\mathcal S}_{1,1}=\{\},{\mathcal S}_{2,1}=\{1,2\}$& (11) of \cite{jaafaar}\\ 
$a_1=1,a_2=1$&${\mathcal S}_{1,1}=\{2\},{\mathcal S}_{2,1}=\{1\}$& (12) of \cite{jaafaar}\\
$a_1=2,a_2=1$&${\mathcal S}_{1,1}=\{1,2\},{\mathcal S}_{1,1}=\{\},{\mathcal S}_{2,1}=\{1\}$& (13) of \cite{jaafaar}\\ 
$a_1=1,a_2=1,a_3=1$&${\mathcal S}_{1,1}=\{2,3\},{\mathcal S}_{2,1}=\{1\},{\mathcal S}_{3,1}=\{\}$ &  (14) of \cite{jaafaar}\\
$a_1=1,a_2=1,a_3=1$&${\mathcal S}_{1,1}=\{2\},{\mathcal S}_{2,1}=\{3\},{\mathcal S}_{3,1}=\{1\}$ &  (15) of \cite{jaafaar}\\
$a_1=1,a_2=1,a_3=1$&${\mathcal S}_{1,1}=\{\},{\mathcal S}_{2,1}=\{2,3\},{\mathcal S}_{3,1}=\{1\}$ &  (16) of \cite{jaafaar}\\
$a_1=1,a_2=1,a_3=1$&${\mathcal S}_{1,1}=\{\},{\mathcal S}_{2,1}=\{\},{\mathcal S}_{3,1}=\{1,2,3\}$ &  (17) of \cite{jaafaar}\\

$a_1=2,a_2=1,a_3=1$&${\mathcal S}_{1,1}=\{1,2,3\},{\mathcal S}_{1,2}=\{\},{\mathcal S}_{2,1}=\{1\},{\mathcal S}_{3,1}=\{\}$ &  (18) of \cite{jaafaar}\\
$a_1=2,a_2=1,a_3=1$&${\mathcal S}_{1,1}=\{2,3\},{\mathcal S}_{1,2}=\{\},{\mathcal S}_{2,1}=\{1\},{\mathcal S}_{3,1}=\{1\}$ &  (19) of \cite{jaafaar}\\
$a_1=2,a_2=1,a_3=1$&${\mathcal S}_{1,1}=\{2\},{\mathcal S}_{1,2}=\{1,3\},{\mathcal S}_{2,1}=\{\},{\mathcal S}_{3,1}=\{1\}$ &  (20) of \cite{jaafaar}\\
$a_1=2,a_2=1,a_3=1$&${\mathcal S}_{1,1}=\{\},{\mathcal S}_{1,2}=\{1,2\},{\mathcal S}_{2,1}=\{3\},{\mathcal S}_{3,1}=\{1\}$ &  (21) of \cite{jaafaar}\\
$a_1=2,a_2=1,a_3=1$&${\mathcal S}_{1,1}=\{\},{\mathcal S}_{1,2}=\{1,2\},{\mathcal S}_{2,1}=\{1,3\},{\mathcal S}_{3,1}=\{\}$ &  (22) of \cite{jaafaar}\\
$a_1=2,a_2=1,a_3=1$&${\mathcal S}_{1,1}=\{\},{\mathcal S}_{1,2}=\{3\},{\mathcal S}_{2,1}=\{1\},{\mathcal S}_{3,1}=\{1,2\}$ &  (23) of \cite{jaafaar}\\
$a_1=2,a_2=1,a_3=1$&${\mathcal S}_{1,1}=\{\},{\mathcal S}_{1,2}=\{\},{\mathcal S}_{2,1}=\{1,2,3\},{\mathcal S}_{3,1}=\{1\}$ &  (24) of \cite{jaafaar}\\
$a_1=2,a_2=1,a_3=1$&${\mathcal S}_{1,1}=\{\},{\mathcal S}_{1,2}=\{\},{\mathcal S}_{2,1}=\{1,3\},{\mathcal S}_{3,1}=\{1,2\}$ &  (25) of \cite{jaafaar}\\

$a_1=3,a_2=1,a_3=1$&${\mathcal S}_{1,1}=\{\},{\mathcal S}_{1,2}=\{\},{\mathcal S}_{1,3}=\{1,2,3\},{\mathcal S}_{2,1}=\{1\},{\mathcal S}_{3,1}=\{1\}$ &  (26) of \cite{jaafaar}\\
$a_1=3,a_2=1,a_3=1$&${\mathcal S}_{1,1}=\{\},{\mathcal S}_{1,2}=\{\},{\mathcal S}_{1,3}=\{1,3\},{\mathcal S}_{2,1}=\{1\},{\mathcal S}_{3,1}=\{1,2\}$ &  (27) of \cite{jaafaar}\\

$a_1=2,a_2=2,a_3=1$&${\mathcal S}_{1,1}=\{1,2,3\},{\mathcal S}_{1,2}=\{2\},{\mathcal S}_{2,1}=\{\},{\mathcal S}_{2,2}=\{\},{\mathcal S}_{3,1}=\{1\}$ &  (28) of \cite{jaafaar}\\
$a_1=2,a_2=2,a_3=1$&${\mathcal S}_{1,1}=\{1,2,3\},{\mathcal S}_{1,2}=\{\},{\mathcal S}_{2,1}=\{\},{\mathcal S}_{2,2}=\{1\},{\mathcal S}_{3,1}=\{2\}$ &  (29) of \cite{jaafaar}\\
$a_1=2,a_2=2,a_3=1$&${\mathcal S}_{1,1}=\{1,2,3\},{\mathcal S}_{1,2}=\{\},{\mathcal S}_{2,1}=\{\},{\mathcal S}_{2,2}=\{\},{\mathcal S}_{3,1}=\{1,2\}$ &  (30) of \cite{jaafaar}\\
$a_1=2,a_2=2,a_3=1$&${\mathcal S}_{1,1}=\{2,3\},{\mathcal S}_{1,2}=\{\},{\mathcal S}_{2,1}=\{\},{\mathcal S}_{2,2}=\{1\},{\mathcal S}_{3,1}=\{1,2\}$ &  (31) of \cite{jaafaar}\\
$a_1=2,a_2=2,a_3=1$&${\mathcal S}_{1,1}=\{2\},{\mathcal S}_{1,2}=\{2\},{\mathcal S}_{2,1}=\{\},{\mathcal S}_{2,2}=\{1,3\},{\mathcal S}_{3,1}=\{1\}$ &  (32) of \cite{jaafaar}\\

$a_1=3,a_2=2,a_3=1$&${\mathcal S}_{1,1}=\{\},{\mathcal S}_{1,2}=\{\},{\mathcal S}_{1,3}=\{1,2,3\},{\mathcal S}_{2,1}=\{\},{\mathcal S}_{2,2}=\{1\},{\mathcal S}_{3,1}=\{1,2\}$ &  (33) of \cite{jaafaar}\\
$a_1=3,a_2=2,a_3=1$&${\mathcal S}_{1,1}=\{\},{\mathcal S}_{1,2}=\{\},{\mathcal S}_{1,3}=\{1,2,3\},{\mathcal S}_{2,1}=\{1\},{\mathcal S}_{2,2}=\{1\},{\mathcal S}_{3,1}=\{2\}$ &  (34) of \cite{jaafaar}\\
$a_1=3,a_2=2,a_3=1$&${\mathcal S}_{1,1}=\{\},{\mathcal S}_{1,2}=\{\},{\mathcal S}_{1,3}=\{2,3\},{\mathcal S}_{2,1}=\{1\},{\mathcal S}_{2,2}=\{1\},{\mathcal S}_{3,1}=\{1,2\}$ &  (35) of \cite{jaafaar}\\
$a_1=3,a_2=2,a_3=1$&${\mathcal S}_{1,1}=\{\},{\mathcal S}_{1,2}=\{\},{\mathcal S}_{1,3}=\{\},{\mathcal S}_{2,1}=\{1\},{\mathcal S}_{2,2}=\{1,2,3\},{\mathcal S}_{3,1}=\{1,2\}$ &  (36) of \cite{jaafaar}\\

$a_1=4,a_2=2,a_3=1$&${\mathcal S}_{1,1}=\{\},{\mathcal S}_{1,2}=\{\},{\mathcal S}_{1,3}=\{\},{\mathcal S}_{1,4}=\{1,2,3\},{\mathcal S}_{2,1}=\{1\},{\mathcal S}_{2,2}=\{1\},{\mathcal S}_{3,1}=\{1,2\}$ &  (37) of \cite{jaafaar}\\

\hline 
\end{tabular}
\label{table-3user} 
\caption{Table II: $a_i$'s and ${\mathcal S}_{i,j}$'s for the case $K=3$}
\end{table}
\twocolumn

\end{document}